\documentclass[12pt,a4paper]{article}

\usepackage{amsmath,amssymb}
\usepackage{graphicx}
\usepackage[nosort]{cite}

\usepackage{tikz}
\usetikzlibrary{decorations.markings}  
\usepackage{epsfig}
\usepackage{enumerate}
\newcommand{\mathsym}[1]{{}}

\usepackage{amsthm,amsbsy,color,multicol}
\usepackage{array,calc}    
\usepackage{rotating}
\usepackage{a4wide,bm}
\usepackage[a4paper,text={450pt,650pt},centering]{geometry}

\let\oldbfseries=\bfseries
\let\oldmdseries=\mdseries
\let\oldnormalfont=\normalfont
\renewcommand{\bfseries}{\oldbfseries\boldmath}
\renewcommand{\mdseries}{\oldmdseries\unboldmath}
\renewcommand{\normalfont}{\oldnormalfont\unboldmath}

\allowdisplaybreaks[3]

\numberwithin{equation}{section}

\usepackage[font=small,labelfont=bf,width=0.85\textwidth]{caption}

\ifx\hypersetup\sadfkjashdfkxja\newcommand\hypersetup[1]{}\fi

\hypersetup{plainpages=false}
\hypersetup{pdfpagemode=UseNone}
\hypersetup{bookmarksnumbered=true}
\hypersetup{pdfstartview=FitH}
\hypersetup{colorlinks=false}
\hypersetup{citebordercolor={.5 1 .5}}
\hypersetup{urlbordercolor={.5 1 1}}
\hypersetup{linkbordercolor={1 .7 .7}}

\DeclareMathSymbol{\Gamma}{\mathalpha}{letters}{"00}
\DeclareMathSymbol{\Delta}{\mathalpha}{letters}{"01}
\DeclareMathSymbol{\Theta}{\mathalpha}{letters}{"02}
\DeclareMathSymbol{\Lambda}{\mathalpha}{letters}{"03}
\DeclareMathSymbol{\Xi}{\mathalpha}{letters}{"04}
\DeclareMathSymbol{\Pi}{\mathalpha}{letters}{"05}
\DeclareMathSymbol{\Sigma}{\mathalpha}{letters}{"06}
\DeclareMathSymbol{\Upsilon}{\mathalpha}{letters}{"07}
\DeclareMathSymbol{\Phi}{\mathalpha}{letters}{"08}
\DeclareMathSymbol{\Psi}{\mathalpha}{letters}{"09}
\DeclareMathSymbol{\Omega}{\mathalpha}{letters}{"0A}

\newcommand{\dd}{\mathrm{d}}
\newcommand{\ii}{\mathrm{i}}

\ifx\genfrac\sdflkaj\else\fi

\newcommand{\set}[1]{\{#1\}}

\newcommand{\beq}{\begin{equation}}
\newcommand{\eeq}{\end{equation}}

\def\[{\begin{equation}}
\def\]{\end{equation}}
\def\<{\begin{eqnarray}}
\def\>{\end{eqnarray}}

\makeatletter
\def\mr@ignsp#1 {\ifx\:#1\@empty\else #1\expandafter\mr@ignsp\fi}%
\newcommand{\multiref}[1]{\begingroup
\xdef\mr@no@sparg{\expandafter\mr@ignsp#1 \: }%
\def\mr@comma{}%
\@for\mr@refs:=\mr@no@sparg\do{\mr@comma\def\mr@comma{,}\ref{\mr@refs}}%
\endgroup}
\makeatother

\newcommand{\hypref}[2]{\ifx\href\asklfhas #2\else\href{#1}{#2}\fi}
\newcommand{\Secref}[1]{Section~\multiref{#1}}

\newcommand{\Figref}[1]{Figure~\multiref{#1}}
\newcommand{\figref}[1]{Fig.~\multiref{#1}}
\renewcommand{\eqref}[1]{(\multiref{#1})}

\makeatletter
\newlength{\apb@width}
\newcommand{\autoparbox}[2][c]{\settowidth{\apb@width}{#2}\parbox[#1]{\apb@width}{#2}}

\makeatother

\ifx\href\asklfhas\newcommand{\href}[2]{#2}\fi

\newtheorem{theorem}{Theorem}
\newtheorem{lemma}{Lemma} 
\newtheorem{remark}{Remark}
\newtheorem{proposition}{Proposition}
\newtheorem{example}{Example}

\newtheorem{definition}{Definition}
\newcommand{\N}{\mathbb{N}} 
\newcommand{\Z}{\mathbb{Z}} 
\newcommand{\vc}[1]{\mathbf{#1}}

\def\cE{{\cal E}} 
 
\def\Lam{\Lambda} 
\def\cP{{\cal P}} 
\def\cI{{\cal I}}

\def\tp{t'}
 
\def\CT{\mbox{CT}} 
\def\yp{y'} 
\def\ypp{y''} 
\def\a{\alpha} 
\def\l{\lambda} 

\def\s{\sigma}

\def\sb{\bar{\sigma}}

\def\om{\omega}
\def\b{\beta}
\def\o{\omega}

\def\d{\delta}

\begin{document}

\renewcommand{\thefootnote}{\fnsymbol{footnote}}
\thispagestyle{empty}
\begin{flushright}\footnotesize
\end{flushright}
\vspace{1cm}

\begin{center}%
{\Large\bfseries%
\hypersetup{pdftitle={Constant term solution for arbitrary number of osculating lattice paths}}%
Constant term solution for an arbitrary \\ number of  osculating lattice paths%
\par} \vspace{2cm}%

\textsc{Richard Brak$^a$ and Wellington Galleas$^{b}$}\vspace{5mm}%
\hypersetup{pdfauthor={Richard Brak, Wellington Galleas}}%

$^a$ \textit{Department of Mathematics \\%
The University of Melbourne\\%
Parkville, VIC 3052, Australia}\vspace{3mm}%


$^b$\textit{Institute for Theoretical Physics and Spinoza Institute, \\ Utrecht University, Leuvenlaan 4,
3584 CE Utrecht, \\ The Netherlands}\vspace{3mm}%

\verb+r.brak@ms.unimelb.edu.au+, %
\verb+w.galleas@uu.nl+ %

\par\vspace{2.5cm}

\textbf{Abstract}\vspace{7mm}

\begin{minipage}{12.7cm}
Osculating paths are sets of directed lattice paths which are 
not allowed to cross each other or have common edges, but are allowed
to have common vertices. In this work we derive a constant term formula for the number of 
such lattice paths by solving a set of simultaneous difference equations.

\hypersetup{pdfkeywords={Osculating paths, Alternating sign matrices, Bethe ansatz}}%
\hypersetup{pdfsubject={}}%
\end{minipage}
\vskip 1.5cm
{\small Mathematics Subject Classifications: 05A15 }
\vskip 0.1cm
{\small Keywords: Osculating paths, Constant term,  Bethe ansatz}
\vskip 1.5cm
{\small June 2013}

\end{center}

\newpage
\renewcommand{\thefootnote}{\arabic{footnote}}
\setcounter{footnote}{0}


\section{Introduction}
\label{sec:intro}

Enumerative combinatorics is basically concerned with the problem of counting
configurations of objects under specified restrictions. Some configurations 
such as the number of combinations of $m$ objects taken $n$ at a time
can be easily counted, but the enumeration of certain kinds of configurations 
are highly non-trivial problems. In particular, this is the case for systems formed by
interacting objects such as vicious and osculating walkers.

Lattice paths generated by vicious and osculating walkers have attracted
a lot of interest over the last decades both in combinatorics and statistical mechanics.
For instance, vicious lattice paths are known to be related to combinatorial objects such as
plane partitions \cite{gessel89, stembridge90, stembridge95}, Young tableaux 
\cite{guttmann98, guttmann00, guttmann03} and symmetric functions \cite{brenti93},
just to name a few connections. From the physical perspective, vicious lattice paths
are also known to offer a good description of polymers \cite{fisher84}. 

Osculating lattice paths in their turn are useful for the description of polymers 
collapse transition \cite{essam03}, but they are also able to describe objects of purely
combinatorial interest. They have been introduced in \cite{brak:1997wo} and are also well known
to be intimately associated  with the combinatorial problem of enumerating alternating sign matrices (ASM).
Alternating sign matrices are square $n\times n$ matrices
whose entries are either $0$, $+1$ or $-1$  such that the non-zero elements in each row and column alternate
between $+1$ and $-1$ and begin and end with $+1$. The total number of $n\times n$ ASM was firstly conjectured by
Robbins, Rumsey and Mills \cite{mills1983a-r,robbins1991a-r} and subsequently proved by
Zeilberger \cite{Zeilberger96} who related it to a particular class of plane partitions. 
These partitions had been enumerated by Andrews \cite{andrews94} based on a result
previously obtained by Stembridge \cite{stembridge:1995rt}.
A shorter derivation was subsequently obtained by Kuperberg \cite{kuperberg95} using the results of Izergin \cite{izergin87} 
and Korepin \cite{korepin82}. We also remark here that another proof based on a formula counting the number of
particular monotones triangles is also available \cite{fischer05}.

In this work we consider the problem of enumerating osculating lattice paths for an arbitrary number of
osculating walkers by establishing a set of partial difference equations counting the number of configurations.
This method has been previously discussed in \cite{brak99} and the solution is obtained 
through a modified version of the celebrated Bethe ansatz \cite{bethe31}.
Although the Bethe ansatz was initially proposed in the study of spin chains, it is 
worth mentioning that the ideas behind it have also been applied in a variety
of contexts. For instance, in the case of the Asymmetric Simple Exclusion Process,
the Bethe ansatz method has resulted in an integral formula for its probabilities \cite{Tracy:2008fk}. 
In our case, however, the solution assumes the form of a constant term formula which allows for a straightforward evaluation.

This paper is organised as follows. In \Secref{sec:osc} we describe the problem of 
osculating lattice paths and establish the conventions used throughout this work. 
In \Secref{sec:proofs} we describe the enumeration problem in terms of partial difference equations
and also present its solution.The \Secref{sec:conclusions} is left for concluding remarks.

\section{Osculating paths}
\label{sec:osc}

Let $S_{n}$ be the group of permutations of $n$ objects with 
$\s=\s (1) \s (2) \ldots \s (n) \in S_{n}$ and let $\sb$ be the inverse of 
$\s$.  We use  the standard notation: $\Z$ is the set of integers, $\N$ is the set of positive integers and $[k]=\set{1,2,\dots k}$. 
Our constant term solution will be intimately connected to the inversion set $C_{\s}$ of a permutation $\s$ which is defined as
\begin{equation}\label{eq:inv}
 C_{\s}=\set{(\a,\b)\in [n]\times[n] \, :\,\text{$\a<\b$ and $\s (\a) >\s (\b)$ }} \; .
\end{equation}
Lattice paths and its osculating case are then defined as follows.
\begin{definition}[Lattice Path]\label{def:lpat}
A lattice path $p$ of length $t\in\N$ on $\Pi=\Z\times\Z$  is a sequence of vertices $v_{0} v_{1} \ldots v_{t}$, 
with $v_{i}\in \Pi$ and  $v_{i}-v_{i-1}\in \set{(1,-1),(1,1)}$ for all  $i\in \set{1,2\dots,t}$. 
If $v_{i}-v_{i-1}=(1,1)$, the step is called an ``up'' step and if $v_{i}-v_{i-1}=(1,-1)$, the step is called an ``down'' step. 
The height  of a vertex $v=(x,y)$ is the value $y$. For a particular path $p$ we denote the corresponding sequence of steps
by $e_{1}e_{2}\ldots e_{t}$ with $e_{i}=(v_{i-1},v_{i})$ for all  $i\in \set{1,2\dots,t}$. 
The height of a step is the height of its left vertex. 
\end{definition}

\begin{definition}[Osculating and Non-intersecting Sequences] 
Let $\{y_{\a}\}_{\a=1\ldots n}$ be a sequence of integers with $y_{1}< y_{2}<\dots< y_{n}$.   
Such a sequence is called a non-intersecting sequence. On the other hand, a sequence of integers with 
$y_{1}\le y_{2}\le\dots\le y_{n}$ such that no three consecutive values are equal,
i.e. if $y_{\a}=y_{\a+1}$ then $y_{\a-1}<y_{a}$ and $y_{\a+1}<y_{\a+2}$,  is called an \textbf{osculating sequence} 
and a pair for which $y_{\a}=y_{\a+1}$ is called an \textbf{osculation}.  
\end{definition} 

The combination of the above definitions allows us to define osculating lattice paths as follows.
We consider lattice paths starting at heights with the same parity in order to prevent paths from stepping across
each other. Without loss of generality we can assume the initial heights to have even parity and
that they are non-intersecting. The parity of the ending heights must then be the same as the parity of the number
of steps and also non-intersecting. These considerations lead to the following definition.

\begin{definition}[Osculating paths]
\label{def_oscPaths}
Let $A_{\alpha}=(0,\yp_{\alpha})$ and $B_{\alpha}=(t,y_{\alpha})$ with $t\in \N$ be the starting and ending
vertices respectively of the $n$-tuple $(w_{1},\ldots,w_{n})$ of lattice paths in $\Pi$ such that
the following conditions hold for all $\a\in \set{1,2,\dots,n}$: 
\begin{enumerate} 
	\item The integers $y'_{\a}$ such that $y'_{\a}< y'_{\a+1}$ have even parity for all $\a\in[n-1]$.
	
	\item The integers $y_{\a}$ such that $y_{\a}< y_{\a+1}$ have the same parity as $t$ for all $\a\in[n-1]$.

	\item $w_{\alpha}$ is a $t$-step path from $A_{\alpha}$ to $B_{\alpha}$ for all $\a\in[n]$.
		
	\item The set $\{\ypp_{\a}\}_{\a=1\ldots n}$ is an osculating sequence  for $0<\tp<t$,  if $s_{\tp}=(\tp,\ypp_{\a})\in w_{\a}$.

	\item The paths $(w_{1},\ldots,w_{n})$ have no steps in common. 
\end{enumerate} 
Paths satisfying the above conditions are called $t$-step osculating paths starting at\\ $(A_{1},\dots, A_n)$ and ending at $(B_{1},\dots, B_n)$. 
\end{definition} 

\begin{figure}[htbp] \centering
\includegraphics{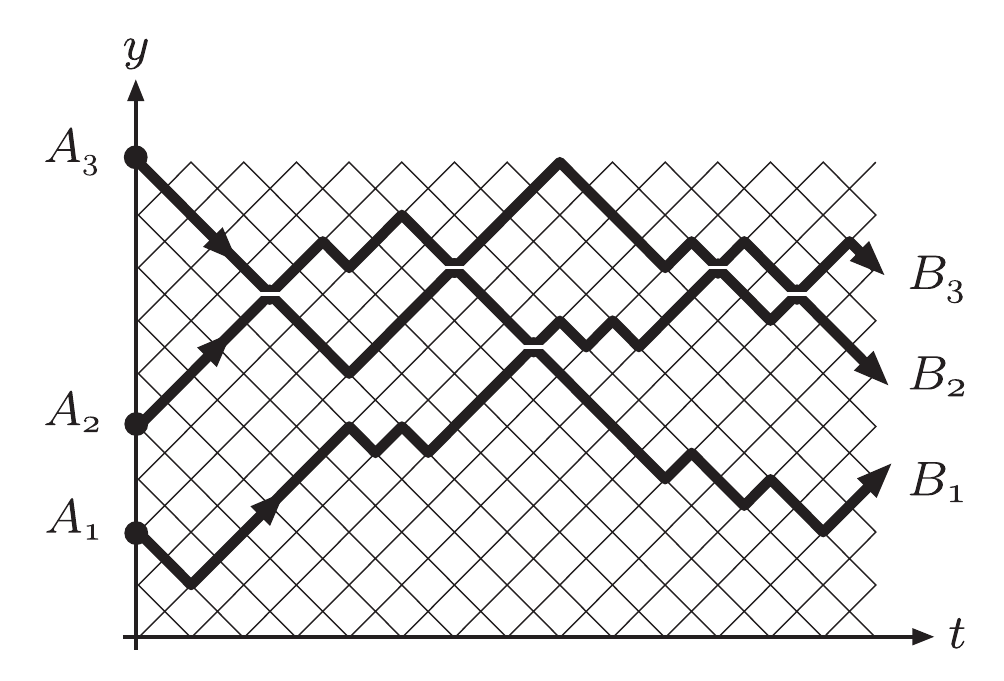}
\caption{An example of	three osculating paths.}
\label{fig:f1}
\end{figure}

For illustrative purposes, we give in \Figref{fig:f1} an example of three osculating paths.
Now in order to define the constant term operation we first need to define a variety of algebraic objects.
Let $\Z[x_{1} ,x_{2}, \ldots,x_{n}]$ be a ring of  polynomials in $x_1,\dots, x_n$ with coefficients
in $\Z$, which we will denote $\Z[x]$, and let $\Z[\om][x]$ be a the ring of polynomials in $x_1,\dots, x_n$
with coefficients in $\Z[\om]$. The corresponding Laurent polynomial rings are then $\Z[x, 1/x]$
and $\Z[\om][x, 1/x]$. In its turn the field of rational functions, i.e. ratios of polynomials
in $x_1,\dots, x_n$ with coefficients in $\Z[\om]$, will be denoted by $\Z[\om](x)$ and the constant term operation
is defined as follows.

\begin{definition}[Constant term] 
Let $R \in \Z[\om](x)$, then the constant term operation $\CT$ is defined as the iterated
contour integrals
\begin{equation}
\label{ctint}
\CT\bigl[R\bigr] =  \frac{1}{(2 \ii \pi)^n} \oint \frac{\dd x_n}{x_n} \left( \oint \frac{\dd x_{n-1}}{x_{n-1}} \left( \dots \left( \oint R \frac{\dd x_1}{x_1} \right) \dots \right)  \right) \; ,
\end{equation}
where the integration contours enclose the origin. 
\end{definition}

Now that we have defined osculating lattice paths and the constant term operation, we can state the main result 
of this paper.
\begin{theorem}
\label{thm1} Let $\l_{\a}=x_{\a}+1/x_{\a}$ and 
$\Lambda_{n}=\prod_{\a=1}^{n}\l_{\a}$. The total number of osculating sequences
for  $t$-step 
osculating paths starting at 
$\{A_{\a}\}_{n}$ and ending at $\{B_{\a}\}_{n}$  is given by
\begin{equation} 
\mathcal{R}_{t}(\om)=\CT \left[ \Lambda_{n}^{t}  \sum_{c_{\vc{\chi}}} \sum_{\sigma\in S_{n}} c_{\vc{\chi}}
\prod_{\alpha=1}^{n} x_{\alpha}^{ \chi_{\a} (y_{\bar{\s} (\a)}-\yp_{\a}) } \,\,
\prod_{( i,j ) \in C_{\sigma}} \left\{ - \frac{ \lambda_{i}\lambda_{j}-\o
x_{j}^{\chi_{j}} /x_{i}^{\chi_{i}}  } {\lambda_{i}\lambda_{j}-\o
x_{i}^{\chi_{i}} /x_{j}^{\chi_{j}}   } \right\}\right]
\label{eq1a}
\end{equation}
where $\chi_{\a} = \pm 1$, $C_{\s}$ is the set of inversions of $\s$, and the coefficients $c_{\vc{\chi}}$
are given by
\begin{equation} \label{ex1}
c_{\vc{\chi}} = \begin{cases}
1 \quad \quad \mbox{if} \;\; \vc{\chi} = ( -1, \dots , -1, \chi_{\a} , -1, \dots , -1 ) : \chi_{\a} = +1 , \; 1 \leq \a \leq \frac{n+1}{2} \cr
-1 \quad \; \mbox{if} \;\; \vc{\chi} = ( -1, \dots , -1, \chi_{\a} , -1, \dots , -1 ) : \chi_{\a} = +1 , \; \frac{n+1}{2} < \a \leq n \cr
0 \qquad \qquad \mbox{otherwise} 
\end{cases}
\end{equation}
for $n$ odd while
\begin{equation} \label{ex2}
c_{\vc{\chi}} = \begin{cases}
1 \quad \quad \mbox{if} \;\; \vc{\chi} = ( +1, -1, \dots , -1, \chi_{\a} , -1, \dots , -1 ) : \chi_{\a} = +1 , \; 2 \leq \a \leq \frac{n}{2}+1 \cr
-1 \quad \; \mbox{if} \;\; \vc{\chi} = ( -1, \dots , -1, \chi_{\a} , -1, \dots , -1 ) : \chi_{\a} = +1 , \; 1 \leq \a \leq \frac{n}{2}-1 \cr
0 \qquad \qquad \mbox{otherwise}  
\end{cases}
\end{equation}
for $n$ even. The variable $\o$ counts the number of osculations.
\end{theorem} 

\begin{example} Using formulaes (\ref{eq1a})-(\ref{ex2}) we find the following polynomials $\mathcal{R}_{t}(\om)$
for $t=2n$ and $y_{\alpha} = y'_{\alpha} = 2(\alpha -1)$.
\begin{itemize}
\item $n=2$: 
\[
\mathcal{R}_{4}(\om) = 20 + 8\om + \om^2
\]
\item $n=3$: 
\[
\mathcal{R}_{6}(\om) = 980 + 1260 \om + 656 \om^2 + 160 \om^3 + 22 \om^4 + 2 \om^5
\]
\item $n=4$: 
\<
\mathcal{R}_{8}(\om) &=& 232848 + 620928 \om + 733824 \om^2 + 499272 \om^3 + 217128 \om^4 \nonumber \\
&&+ \; 64876 \om^5 + 13657 \om^6 + 1974 \om^7 + 189 \om^8 + 18 \om^9 + \om^{10}
\>
\item $n=5$: 
\<
\mathcal{R}_{10}(\om) &=& 267227532 + 1214670600\om + 2549915280\om^2 + 3274813212\om^3 \nonumber \\
&&+ \; 2879827684\om^4 + 1844895472\om^5 + 895616536\om^6 + 337943000\om^7 \nonumber \\
&&+ \; 100663338\om^8 + 23882812\om^9 + 4536546\om^{10} + 694008\om^{11} \nonumber \\
&&+ \; 83888\om^{12} + 7892\om^{13} + 604\om^{14} + 46\om^{15} + 2\om^{16}
\>
\end{itemize}
\end{example}

\section{Partial difference equations approach}
\label{sec:proofs}   

The total number of osculating sequences given in Theorem $1$ 
satisfies a partial first order difference equation, in addition
to an osculation constraint and an initial condition. The osculation
constraint is also given by first order difference relations. 
In particular, if we have $n$ paths then the number of recurrence relations
associated to the osculation process equals the Fibonacci number $F_{n}$. 
In this way the proof of Theorem $1$ will rest on the exact solution of the
aforementioned conditions.

\paragraph{Matchings and osculating sequences.} We begin by constructing a matching
from an osculating sequence in the sense of graph theory. Let $Z_{n}$ be a linear graph
of $n$  vertices and let $\cP$ be the set of all matchings of $Z_{n}$. 
Also let $\{y_{\a}\}_{\a=1\ldots n}$ be an osculating sequence. Next we label the
$\a^{\text{th}}$ vertex of $Z_{n}$ by the corresponding variable $y_{\a}$ 
and colour the edge between $\a$ and $\a+1$ for all pairs such that
$y_{\a}=y_{\a+1}$. This defines a unique matching $M\in \cP$.  Now let $\cI_{M}$ be the set of 
isolated points in $M$ and $\cE_{M}$ the set of coloured edges in $M$.  
For a given configuration of osculating paths and a given horizontal 
coordinate we get a set of height coordinates $(y_1,\ldots,y_n)$.  
This set naturally defines an osculating sequence $\{y_{\a}\}_{\a=1\ldots n}$ according to Definition 2,
and an example of a matching is shown in \figref{fig:f2}.
\begin{figure}[htbp] 
\centering
\includegraphics{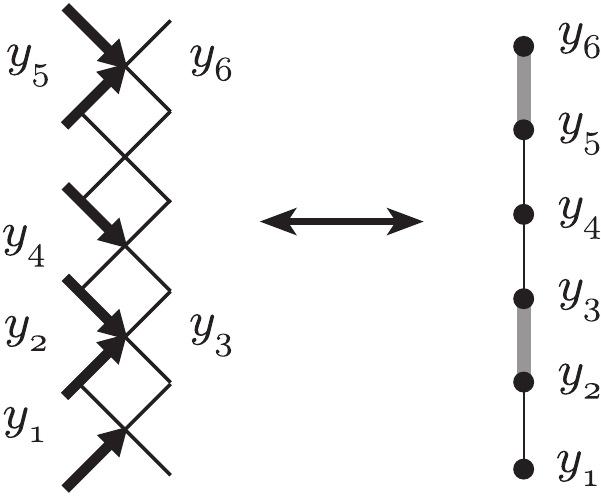}
\caption{An example of the matching obtained from a set of osculating 
vertices. Here $y_{2}=y_{3}$, $y_{5}=y_{6}$, $\cI_{M}=\{y_{1},y_{4}\}$ and $\cE_{M}=\{[y_{2},y_{3}], [y_{5},y_{6}]\}$.}
\label{fig:f2}
\end{figure}

\paragraph{Difference equations.} Let $r(\vc{y};t)$ be the osculation generating
function for the number of $t$-step osculating paths from $\{A_{\a}\}_{n}$ to $\{B_{\a}\}_{n}$
with $\vc{y}=(y_1,\ldots,y_n)$ and $\vc{e}=(e_1,\ldots,e_n)$. If \mbox{$y_{1}<y_{2}< \ldots <y_{n}$} 
then \begin{equation}
r(\vc{y};t+1)=  \sum_{e_1=\pm 1}\cdots
\sum_{e_n=\pm 1} 
r(\vc{y}+\vc{e};t).
\label{r1}
\end{equation}

\paragraph{Osculation constraint.} The osculation process is characterised 
by another difference relation in addition to (\ref{r1}). 
Since $\vc{y}$ defines an osculating sequence
$\{y_{\a}\}_{\a=1\ldots n}$, and hence a matching, we have the condition
\begin{equation}
r(\vc{y};t+1)=\o^{|\cE_{M}|} \sum_{e_1\in v^{M}_1}\cdots
\sum_{e_n\in v^{M}_n} 
r(\vc{y}+\vc{e};t)
\label{r2}
\end{equation} 
for each $M\in \cP$ where
\begin{equation}
 v^{M}_{i}= 
 \begin{cases} 
 \{+1,-1\} & \text{if $y_{i}\in\cI_{M}$}, \\
 \{-1\} & \text{if  $y_{i}$ is a lower vertex of some 
 edge $\in \cE_{M}$ },\\
    \{+1\} & \text{if  $y_{i}$ is an upper vertex of some 
 edge $\in \cE_{M}$. }
 \end{cases}
 	\label{r3}
\end{equation}
The upper vertex of an edge in $\cE_{M}$ is the one with the 
greater $\a$ label and conversely for the lower vertex. 
Thus we have a number $F_n$ of equations (\ref{r2}) since $|\cP|=F_{n}$.

\paragraph{Initial condition.} The difference equations (\ref{r1}) and (\ref{r2})
still need to be complemented with an initial condition in order to have $r(\vc{y};t)$
completely determined. 
From the previous discussions we thus have the initial condition
\begin{equation}
r(\vc{y};0)=\prod_{\a=1}^{n}\delta_{y_{\a},\yp_{\a}} 
\label{r4}
\end{equation}
where $y_{\a}$ and $y'_{\a}$ satisfy the conditions of Definition 3. 
In (\ref{r4}) $\delta_{i,j}$ stands for the Kronecker
delta. In what follows we shall present a solution for the relations (\ref{r1})-(\ref{r4})
based on the celebrated Bethe ansatz \cite{bethe31}.

\bigskip

Let us now consider $\vc{x},\vc{x}^{\vc{y}}\in\Z[x,1/x]$ defined as $\vc{x}=\prod_{\a=1}^{n}x_{\a}$ and $\vc{x}^{\vc{y}}
=\prod_{\a=1}^{n}x_{\a}^{y_{\a}}$. Then the trial solution $\Lam_{n}^{t} \vc{x}^{\vc{y}}$ satisfies (\ref{r1})
provided that
\begin{equation}
\Lam_{n}(x_{1},\ldots,x_{n})=\prod_{\a=1}^{n}\bigl(x_{\a}+
x_{\a}^{-1}\bigr).
\label{r5}
\end{equation}
The term $\Lam_{n}(x_{1},\ldots,x_{n})$ is a symmetric function, i.e. it is invariant
under any permutation of its arguments. Thus $\Lam_{n}^{t} \vc{x}_{\s}^{\vc{y}}$
with $\vc{x}_{\s}^{\vc{y}}=\prod_{\a=1}^{n}x_{\s (\a) }^{y_{\a}}$
for any $\s\in S_{n}$ is also a solution of (\ref{r1}). Due to the linearity of
(\ref{r1}) we thus have that 
\begin{equation}   
\psi (\vc{y};t) = \Lam_{n}^{t} \sum_{\s\in S_{n}}A_{\s}(\vc{x}) \vc{x}_{\s}^{\vc{y}}
\label{r5b}
\end{equation}
solves (\ref{r1}) if $A_{\s}(\vc{x})$ is independent of $\vc{y}$.   

Up to this stage the function $A_{\s}(\vc{x})$ is arbitrary. However, we will see that
it can be conveniently written as elements of $\Z[\om](x)$ in order to satisfy (\ref{r2}). 
This will be of importance for the introduction of the constant term operation since this
has been only defined on $\Z[\om](x)$ (\ref{ctint}).

The Eq. (\ref{r2}) is satisfied by $\psi (\vc{y};t)$ as defined by (\ref{r5b}) if 
\begin{equation}
\sum_{\s\in S_{n}} 
A_{\s}\biggl(\Lam_{n}
-\o^{|\cE_{M}|} \sum_{e_1\in v^M_1}\cdots \sum_{e_n\in v^M_n} 
\vc{x}_{\s}^{\vc{e}}\biggr)
\vc{x}_{\s}^{\vc{y}}=0
\label{at}
\end{equation} 
for each $M\in\cP$. The choice of elements $A_{\s}$ ensuring (\ref{at}) is given by the following lemma.
\begin{lemma}
If
\begin{equation}
A_{\s} = \prod_{ ( \alpha , \beta  ) \in C_{\sigma} } s_{\alpha \beta}
\label{r6}
\end{equation}
where
\begin{equation}
s_{\a\b}= - \frac{\l_{\a}\l_{\b}-\o x_{\b}/x_{\a}}{\l_{\a}\l_{\b}-\o x_{\a}/x_{\b}}
\label{r7}
\end{equation}
Then, 
\begin{equation}
 \sum_{\s\in S_{n}} 
 A_{\s}\biggl(\Lam_{n}
 -\o^{|\cE_{M}|} \sum_{e_1\in v^M_1}\cdots \sum_{e_n\in v^M_n} 
\vc{x}_{\s}^{\vc{e}}\biggr)
 \vc{x}_{\s}^{\vc{y}}=0
 \label{r8}
\end{equation} 
holds for arbitrary $\vc{y}$ and every $M\in\cP$. Note that $s_{\a,\b}\in \Z[w](x)$ after multiplying 
numerator and denominator by suitable $x_\a x_\b$ factors.
\end{lemma}
\paragraph{Proof.} The lemma is, although tediously, readily proved by induction on $|\cE_{M}|$.   

\bigskip

Expressions like (\ref{r5b})-(\ref{r8}) are known as Bethe ansatz and they appear in a variety 
of contexts and versions. See for instance \cite{baxter82} for applications of the Bethe ansatz
in Exactly Solvable Models of statistical mechanics. It is also worth to stress here the correspondence
between $\psi (\vc{y};t)$, with $A_{\s}$ given by (\ref{r6}) and (\ref{r7}), 
and Bethe's wave function for the six-vertex model with toroidal boundary conditions \cite{Lieb67,baxter82}. 
Although here the variables $x_\alpha$ are not constrained by Bethe ansatz equations, as it happens
for the six-vertex model, the variable $\omega$ could still be related to the six-vertex model anisotropy parameter
$\Delta$ to find a correspondence between $\psi (\vc{y};0)$ and the six-vertex model wave function. More precisely, if we consider
the conventions of \cite{baxter82} we then have $x^2_{\alpha} = 2 \Delta z_{\alpha} -1$ and $\omega = 4 \Delta^2 - 1$.

The function $\psi (\vc{y};t)$ satisfy the conditions (\ref{r1}) and (\ref{r2}), and for last we need 
to consider the initial condition (\ref{r4}). In the traditional Bethe ansatz technique one would look
for equations constraining the variables $x_{\a}$ for that. Here we find that such approach is not 
suitable and instead we shall consider the constant term operation. Before proceeding with this analysis
we first need to remark a discrete symmetry of (\ref{r5}). 
We notice the function $\Lam_n$ is also invariant under the mapping $x_{\a} \rightarrow \frac{1}{x_{\a}}$
which implies that $\psi (\vc{y};t)$ is still a solution of (\ref{r1}) and (\ref{r2}) under this operation.
Thus we can define a set of functions $\psi_{\chi} (\vc{y};t)$ corresponding to $\psi (\vc{y};t)$ 
with the replacement $x_{\a} \rightarrow x_{\a}^{\chi_{\a}}$ where $\chi_{\a}$ can assume the values $\pm 1$.
This yields $2^n$ solutions of (\ref{r1}) and (\ref{r2}) which can be linearly combined 
to satisfy (\ref{r4}). Moreover, since $\psi_{\chi} (\vc{y};t) \in \Z[\om](x)$, the constant
term $\CT\bigl[ \psi_{\chi} \bigr]$ also satisfies (\ref{r1}) and (\ref{r2}). This can be
readily seen from the integral formula (\ref{ctint}). We have now gathered all
the ingredients to present a solution for the total number of osculating sequences.
\begin{lemma}  
\label{initial-lemma}
Let 
\begin{equation}
r(\vc{y};t)=\CT\Bigl[ \sum_{\vc{\chi}} c_{\vc{\chi}} \; \vc{x}^{- \vc{\chi} \cdot \vc{y'}} \Lam_n^t
\sum_{\s\in S_{n}} A_{\s}^{\vc{\chi}} \; \vc{x}_{\s}^{\vc{\chi} \cdot \vc{y}} \Bigr]
\label{r5c}
\end{equation}
where $\vc{x}^{- \vc{\chi} \cdot \vc{y'}} = \prod_{\a=1}^{n} x_{\a}^{- \chi_{\a} y'_{\a}}$,
$\vc{x}_{\s}^{\vc{\chi} \cdot \vc{y}} = \prod_{\a=1}^{n} x_{\s(\a)}^{ \chi_{\s(\a)} y_{\a}}$, $\chi_{\a} = \pm 1$ and
\begin{equation}
A_{\s}^{\vc{\chi}} = \prod_{ ( \alpha , \beta  ) \in C_{\sigma} } s_{\alpha \beta}^{\vc{\chi}}
\label{r5d}
\end{equation}
with
\begin{equation}
s_{\a\b}^{\vc{\chi}} = - \frac{\l_{\a} \l_{\b} -\o x_{\b}^{\chi_{\b}} /x_{\a}^{\chi_{\a}} }{\l_{\a} \l_{\b}  -\o x_{\a}^{\chi_{\a}} /x_{\b}^{\chi_{\b}}} \; .
\label{r5e}
\end{equation}
For $n$ odd and coefficients
\begin{equation}
\label{oddchi}
c_{\vc{\chi}} = \begin{cases}
1 \quad \quad \mbox{if} \;\; \vc{\chi} = ( -1, \dots , -1, \chi_{\a} , -1, \dots , -1 ) : \chi_{\a} = +1 , \; 1 \leq \a \leq \frac{n+1}{2} \cr
-1 \quad \; \mbox{if} \;\; \vc{\chi} = ( -1, \dots , -1, \chi_{\a} , -1, \dots , -1 ) : \chi_{\a} = +1 , \; \frac{n+1}{2} < \a \leq n \cr
0 \qquad \qquad \mbox{otherwise} \; ,  
\end{cases}
\end{equation}
the initial condition (\ref{r4}) is satisfied. For $n$ even the initial condition (\ref{r4}) requires
\begin{equation}
\label{evenchi}
c_{\vc{\chi}} = \begin{cases}
1 \quad \quad \mbox{if} \;\; \vc{\chi} = ( +1, -1, \dots , -1, \chi_{\a} , -1, \dots , -1 ) : \chi_{\a} = +1 , \; 2 \leq \a \leq \frac{n}{2}+1 \cr
- 1\quad \; \mbox{if} \;\; \vc{\chi} = ( -1, \dots , -1, \chi_{\a} , -1, \dots , -1 ) : \chi_{\a} = +1 , \; 1 \leq \a \leq \frac{n}{2}-1 \cr
0 \qquad \qquad \mbox{otherwise} \; .  
\end{cases}
\end{equation}
\end{lemma}

\paragraph{Proof.} Let us call $e$ the identity element of $S_n$. That is the element of $S_n$ such
that $\sigma (\a) = \a$. Also let us define the set $\bar{S}_n = S_n \backslash \{ e \} $.
Thus for $t=0$ and considering only the identity element in the sum over $S_n$ 
of expression (\ref{r5c}), we obtain
\begin{equation}
\CT\Bigl[ \sum_{\vc{\chi}} c_{\vc{\chi}} \; \vc{x}^{\vc{\chi} \cdot ( \vc{y} - \vc{y'})} \Bigr] = 
\prod_{\a=1}^{n}\d_{y_{\a},\yp_{\a}} \sum_{\vc{\chi}} c_{\vc{\chi}} \; .
\end{equation}
From (\ref{oddchi}) and (\ref{evenchi}) we have that $\sum_{\vc{\chi}} c_{\vc{\chi}} = 1$ and to prove the lemma we are reduced 
to showing that
\begin{equation}
\label{targ}
\CT\Bigl[ \sum_{\vc{\chi}} c_{\vc{\chi}} \; \vc{x}^{- \vc{\chi} \cdot \vc{y'}}
\sum_{\s\in \bar{S}_{n}} A_{\s}^{\vc{\chi}} \; \vc{x}_{\s}^{\vc{\chi} \cdot \vc{y}} \Bigr] = 0
\end{equation}
for $y_{\a},y'_{\a} \in \mathbb{Z}: y_{\a} < y_{\a+1}, y'_{\a} < y'_{\a+1}$.

\paragraph{Terms with vanishing constant term.} The function $A_{\s}^{\vc{\chi}}$ given 
by (\ref{r5d}) and (\ref{r5e}) can be expanded as
\[
A_{\s}^{\vc{\chi}} = \sum_{m_i \geq 0} \phi_{m_1, \dots , m_n} \prod_{\a =1}^{n} x_{\s (\a)}^{m_{\a}} \; ,
\]
for any configuration $\vc{\chi}$. Thus the term inside the bracket in the LHS of (\ref{targ}) will be of the form
\[
\label{targ1}
\sum_{\vc{\chi}} \sum_{m_i \geq 0} \sum_{\s\in \bar{S}_{n}}  c_{\vc{\chi}} \phi_{m_1, \dots , m_n}^{\vc{\chi}}  \prod_{\a =1}^{n} x_{\a}^{\chi_{\a}(y_{\bar{\s}(\a)} - y'_{\a}  ) + m_{\bar{\s}(\a)}} \; .
\]
For a given configuration $(\s , \vc{\chi})$, the expression (\ref{targ1}) will produce a non-vanishing constant
term only if 
\[
\label{targ2}
\chi_{\a}(y_{\bar{\s}(\a)} - y'_{\a}  ) + m_{\bar{\s}(\a)} = 0 \quad \quad \forall \alpha \; .
\]

\bigskip

\begin{proposition} The configuration $(\s , \vc{\chi})$ such that 
$\exists \; (\a_1 , \a_2) : \a_1 < \a_2 , \bar{\s}(\a_1) > \bar{\s}(\a_2), \chi_{\a_1} = 1, \chi_{\a_2} = -1$
does not produce constant term.
\end{proposition}

\begin{proof} For such configuration the equation (\ref{targ2}) gives us the relations 
\<
\label{targ3}
y_{\bar{\s}(\a_1)} - y'_{\a_1}  + m_{\bar{\s}(\a_1)} &=& 0 \nonumber \\
-y_{\bar{\s}(\a_2)} + y'_{\a_2}  + m_{\bar{\s}(\a_2)} &=& 0 \; ,
\>
which can be summed up yielding the identity
\[
\label{targ4}
( y_{\bar{\s}(\a_1)} - y_{\bar{\s}(\a_2)} ) + (y'_{\a_2} - y'_{\a_1}) + (m_{\bar{\s}(\a_1)} + m_{\bar{\s}(\a_2)} ) = 0 \; .
\]
Now since $m_{\a} \geq 0$, $y_{\a} < y_{\a +1}$ and $y'_{\a} < y'_{\a +1}$, the Eq. (\ref{targ4}) can not be satisfied
and consequently (\ref{targ2}) does not hold.
\end{proof}
\begin{remark} Analogously a configuration $(\s , \vc{\chi})$ such that
$\exists \; (\a_1 , \a_2) : \a_1 > \a_2 , \bar{\s}(\a_1) < \bar{\s}(\a_2), \chi_{\a_1} = -1, \chi_{\a_ 2} = 1$ does
not produce a constant term as well.
\end{remark}

\paragraph{Terms with non-vanishing constant term.} Considering only the non-null coefficients
according to (\ref{oddchi}) and (\ref{evenchi}), the components whose constant term does not vanish in the LHS of (\ref{targ}), reorganise
as 
\[
\label{reorg1}
\sum_{\stackrel{1 \leq \a \leq \frac{n+1}{2}}{\frac{n+1}{2} < \b \leq n }} 
\left( c_{( -1, \dots, -1 , \chi_{\a}, -1, \dots, -1 )} + c_{(-1, \dots, -1 , \chi_{\beta}, -1, \dots, -1 )} \right) \Psi_{\a \b} \qquad \;\; ( \chi_{\a , \b} = +1 )
\]
for $n$ odd and as
\[
\label{reorg2}
\sum_{\stackrel{2 \leq \a \leq \frac{n}{2}+1 }{1 \leq \b \leq \frac{n}{2}-1}} 
\left( c_{( +1, -1, \dots, -1 , \chi_{\a}, -1, \dots, -1 )} + c_{( -1, \dots, -1 , \chi_{\beta}, -1, \dots, -1 )} \right) \bar{\Psi}_{\a \b} \qquad \;\; ( \chi_{\a , \b} = +1 )
\]
for $n$ even. Although it is a lengthy computation, the expressions (\ref{reorg1}) and (\ref{reorg2})
follows from the property $\CT \left[  f(x_{\a}) \right] = \CT \left[  f(1/x_{\a}) \right]$ for any Laurent
polynomial $f$. The form of the functions $\Psi_{\a \b}$ and $\bar{\Psi}_{\a \b}$ will not be required here but they consist
of the explicit evaluation of the constant terms in (\ref{targ}). It is also important to remark here that when
evaluating the constant term on $\Z[\om](x)$ as defined in (\ref{ctint}) using the residue formula, the terms of the form
$(1 - \om + Q)^{-1}$ with $Q \in \Z[x]$ and $\mbox{CT} [Q] = 0$ need to be expanded as $\sum_{n=0}^{\infty} Q^{-n-1} (\om - 1)^n$
in order to ensure that all poles at the origin are being captured by the integration contours. Finally,
we can see that the expressions (\ref{reorg1}) and (\ref{reorg2}) vanish for coefficients $c_{\vc{\chi}}$ respectively 
given by (\ref{oddchi}) and (\ref{evenchi}). This completes our proof.

\section{Concluding remarks}	
\label{sec:conclusions}

The main result of this work is the constant term formula (\ref{eq1a}) counting
the number of lattice paths generated by an arbitrary number of osculating walkers. 
This formula has its origins in a Bethe ansatz like expression but it still contains modifications
from the usual Bethe ansatz. More precisely, the sum over variables $\chi$ present
in (\ref{eq1a}) is a new feature of our solution and it has been introduced in order to fulfil
the initial condition (\ref{r4}).

It is worth remarking here that the case of three osculating walkers had
been previously considered in \cite{bousquet06} through a step by step decomposition
of osculating configurations. Although the method of \cite{bousquet06} can be formally extended
for arbitrary number of osculating walkers, the solution of the obtained equation
seems to be out of reach. 

The list of problems related to osculating lattice paths is still not as abundant
as the case of vicious walkers but new connections have emerged recently. For instance,
in the work \cite{korff12} it was demonstrated that the counting of rational curves intersecting 
Schubert varieties of the Grassmannian are related to the counting of osculating lattice paths
on the cylinder.
Although the enumeration of ASM is well known and three different proofs are available \cite{Zeilberger96, kuperberg95, fischer05}, a purely
combinatorial proof remains an open problem which we hope this work to shed some light upon.

\section*{Acknowledgements}	 
The authors thank the Australian Research Council (ARC) and 
the Centre of Excellence for Mathematics and Statistics of
Complex Systems (MASCOS) for financial support.

\bibliographystyle{plain}

\bibliography{osculatingBibTex}

\end{document}